\documentclass[10pt, conference, letterpaper]{IEEEtran}
\usepackage{amsmath}
\usepackage{amsthm}
\usepackage{amssymb}
\usepackage{graphicx}
\usepackage{subfigure}
\usepackage{caption}
\usepackage{color}
\usepackage{balance}
\usepackage{cite}
\usepackage{url}
\usepackage{enumerate}
\usepackage{makecell}
\newtheorem{definition}{Definition}

\newtheorem{lemma}{Lemma}
\newtheorem{corollary}{Corollary}
\newtheorem{theorem}{Theorem}
\IEEEoverridecommandlockouts
\begin{document}
\title{Performance and Stability of Barrier Mode Parallel Systems with Heterogeneous and Redundant Jobs}
\author{\IEEEauthorblockN{Brenton Walker \quad\quad Markus Fidler\thanks{This work was supported in part by the German Research Council (DFG) under Grant VaMoS (FI 1236/7-1,2). This work expands on the previous conference paper~\cite{bem-infocom}.}}
\IEEEauthorblockA{Institute of Communications Technology, Leibniz Universit\"{a}t Hannover}}
\maketitle
\begin{abstract}
In some models of parallel computation, jobs are split into smaller tasks and can be executed completely asynchronously.  In other situations the parallel tasks have constraints that require them to synchronize their start and possibly departure times.  This is true of many parallelized machine learning workloads, and the popular Apache Spark processing engine has recently added support for Barrier Execution Mode, which allows users to add such barriers to their jobs. These barriers necessarily result in idle periods on some of the workers, which reduces their stability and performance, compared to equivalent workloads with no barriers.

In this paper we will consider and analyze the stability and performance penalties resulting from barriers.  We include an analysis of the stability of $(s,k,l)$ barrier systems that allow jobs to depart after $l$ out of $k$ of their tasks complete. We also derive and evaluate performance bounds for hybrid barrier systems servicing a mix of jobs, both with and without barriers, and with varying degrees of parallelism.  For the purely 1-barrier case we compare the bounds and simulation results to benchmark data from a standalone Spark system.  We study the overhead in the real system, and based on its distribution we attribute it to the dual event and polling-driven mechanism used to schedule barrier-mode jobs.  We develop a model for this type of overhead and validate it against the real system through simulation.
\end{abstract}
%
%
\section{Introduction}
\label{sec:introduction}

\begin{figure}
    \centering
    \subfigure[Single-queue asynchronous task scheduling]{
        \includegraphics[width=2in,height=1.17in]{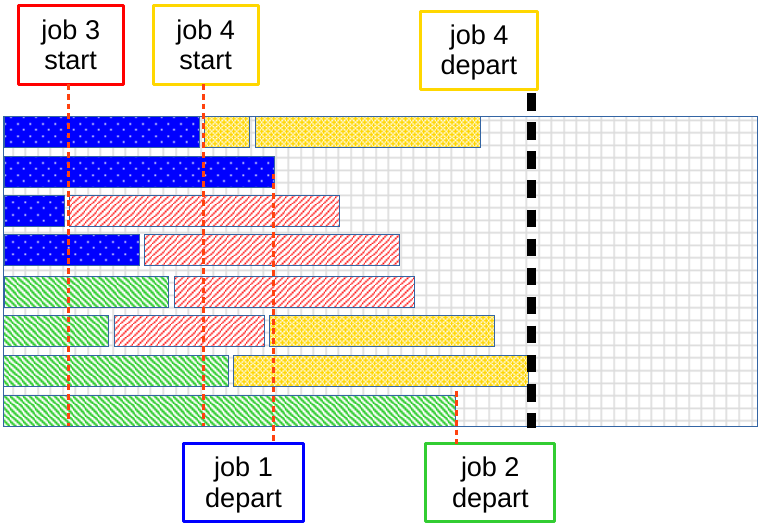}
        \label{fig:nisq-scheduling}
    }
    \subfigure[1-barrier BEM system]{
        \includegraphics[width=2in,height=1.17in]{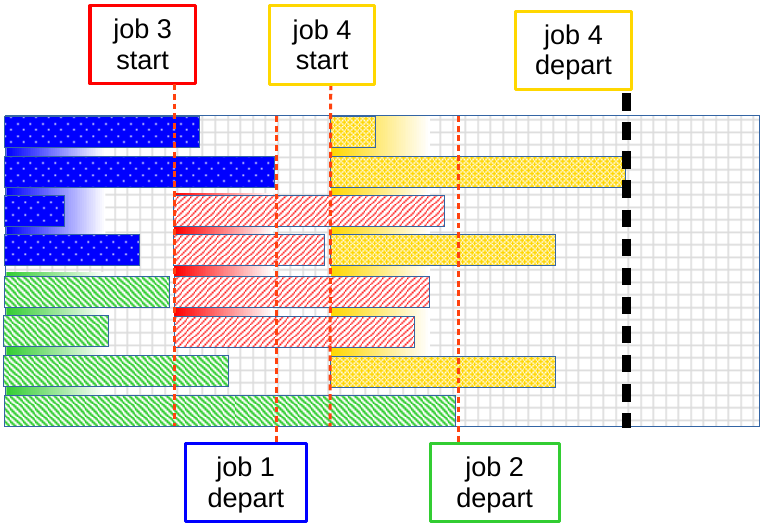}
        \label{fig:bem1-scheduling}
    }
    \subfigure[2-barrier BEM system]{
        \includegraphics[width=2in,height=1.17in]{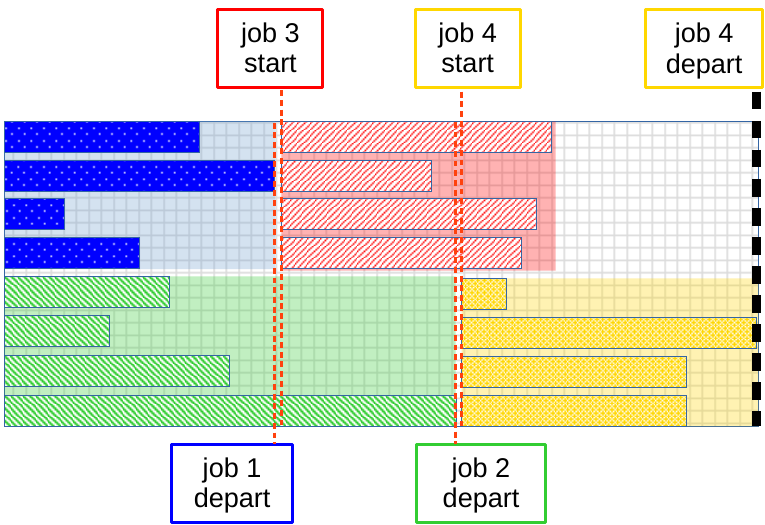}
        \label{fig:bem2-scheduling}
    }
    \caption{Single-queue asynchronous task scheduling vs 1 and 2-barrier BEM systems.
    }
    \label{fig:bem-scheduling}
    \vspace*{-5mm}
\end{figure}

In parallel systems, work arrives in the form of {\bf jobs} which are divided up into $k$ {\bf tasks} and executed in parallel on a collection of $s$ {\bf workers}.  There are many architectures supporting parallelism, and many models used to study them.  Their performance and stability regions depend on the constraints of the system and the variance of the service times of the parallel tasks.  The most fundamental constraint is that a job cannot depart until all (or enough) of its tasks complete.  More specific constraints are whether or not workers can take up new tasks as soon as its previous task is finished ({\bf Fork-Join}~\cite{flatto:forkjoin, nelson:forkjoin} vs {\bf Split-Merge}~\cite{harrison:splitmerge}), and whether tasks are bound to a particular worker upon arrival ({\bf Fork-Join}~\cite{flatto:forkjoin, nelson:forkjoin}), or if all tasks are held in a single queue, and the next task in the queue can be immediately serviced by the next available worker ({\bf Single-Queue Fork-Join}~\cite{nelson:parallelprocessing, walker:icfc2017}).  The last model is the most accurate abstraction for the so-called ``embarrassingly parallel'' map-reduce model supported by popular engines like Spark, Hadoop, and Flink.

With the overwhelming interest in machine learning, however, even systems like these have begun to support constraints known as {\bf barriers}.  A start barrier requires that all tasks begin execution at the same time.  In particular this requires that $k \le s$.  A departure barrier requires that completed tasks block access to their respective workers until the job departs.  Many machine learning workloads have such constraints, as we will discuss in the next section.

Following the terminology used in Apache Spark, we refer to these as {\bf Barrier Execution Mode (BEM)} systems.  They have blocking start barriers, and optionally blocking departure barriers.  We do not assume that the number of tasks equals the number of workers, just that $k \le s$.  If a job with $k$ tasks and a start barrier comes to the front of the queue and fewer than $k$ workers are available, the job must wait until more tasks complete.
A blocking departure barrier requires that the workers cannot service a new task until the entire job of their previous task departs.  We will often refer to these as \mbox{\bf 1-barrier} (start barrier only) and \mbox{\bf 2-barrier} (start and departure barrier) systems in this paper.

Both start and departure barriers necessarily reduce the stability region and performance of parallel systems compared to equivalent systems with no barriers, because up to $k-1$ servers may sit idle while a job is blocked.
The differing behavior of systems with no constraints, one barrier, and two barriers, is illustrated in Fig.~\ref{fig:bem-scheduling} for systems with $k=4$ tasks per job, and $s=8$ workers. Fig.~\ref{fig:bem-scheduling}(a) shows the optimal scheduling of a Single-Queue Fork-Join system with no blocking barriers (conventional Spark). Note that no workers idle as long as there are tasks waiting to be processesed.  Fig.~\ref{fig:bem-scheduling}(b) shows how those same tasks would be scheduled on a BEM system with a blocking start barrier. Note the brief period when three workers are idle before the fourth one becomes available and job~3 can begin service.  Similarly for job~4.  Fig.~\ref{fig:bem-scheduling}(c) shows how the same tasks would be scheduled on a BEM system with blocking start and departure barriers. The start of jobs~3 and~4 are delayed by a larger degree, and there is much more worker idle time.

Note that when $k=s$, that is, each job requires all of the workers to run, both 1 and 2-barrier systems behave equivalently to the classical Split-Merge model.  Due to their synchronous operation, Split-Merge parallel systems are probably the easiest to implement, and have been frequently studied~\cite{harrison:splitmerge, lebrecht:forkjoin, rizk:forkjoin, joshi:knforkjoin}, but they have disastrous scaling performance, and a stability region that shrinks with increasing parallelism.  Letting $k<s$ in BEM systems mitigates this performance penalty to a significant degree.  In more general situations where the number of tasks per job, $k$, varies between jobs, adds another dimension to the performance analysis.

Another strategy to mitigate the performance problems arising from barriers is to abandon straggler tasks.  In an $(s,k,l)$ BEM parallel system , there are $s$ workers, and jobs are divided into $k$ tasks, but a job can depart, and its remaining tasks preempted, when $l\le k$ of its tasks finish.

Finally, given that the presence of BEM jobs causes idle ``bubbles'' in the scheduling of a system, it is natural to wonder how severely the coexistence of BEM and non-BEM jobs on a parallel system affects the stability and performance.

In this paper we will study and quantify the stability and scaling performance of BEM systems.  We show how the stability region of one and two-barrier BEM systems varies with $s$ and $k$, and compute stability bounds for $(s,k,l)$ BEM systems that take into account the fact that some work is discarded.  We find that a 2-barrier $(s,k,l)$ system can exceed the stability region of a conventional BEM system under some circumstances, and that the performance advantage can be greater, depending on the distribution of the task sizes.  For conventional BEM systems with one barrier we use techniques from stochastic network calculus to obtain performance bounds. BEM systems with two barriers are easier to model because they can be represented as M$\mid$G${\mid}m$ queues, see~\cite{bem-infocom}. We show that for a given number of workers, $s$, there is an optimal degree of parallelism, $k$, which depends on $s$ and the system utilization.  Our results are applicable to cases where there are a mix of BEM and non-BEM jobs, and to systems where the jobs' degrees of parallelism, $k$, is variable.  We also carry out benchmarking experiments of Barrier Execution Mode in a standalone Spark cluster to verify the theoretically-predicted performance.  We find that there is a substantial amount of overhead involved in BEM operations that originates from the interplay between polling intervals and job arrivals.  We analyze the overhead distribution and model it in simulation.
%
%
\subsection{Systems with Blocking Barriers}
The primary solution to accelerating the processing of large amounts of data has been to break jobs into many tasks and execute them in parallel.  The constraints under which this parallel execution takes place vary.  In the last decade, the ``embarrassingly parallel'' map-reduce model of parallel execution has become extremely popular, with open-source engines such as Hadoop, Spark, and Flink seeing wide adoption~\cite{hadoop,flink-intro,spark-2016}.  These tools are extremely efficient when the tasks are completely independent, but many algorithms need to be rewritten, or simply cannot be effectively implemented in the map/reduce API.  In particular, any parallel algorithm that requires the Allreduce type of operations in MPI and similar parallel architectures, require task synchronization points. Many machine learning workloads can be effectively parallelized, but in order to ensure that all workers are using the current model state, actually require that all tasks begin simultaneously, and sometimes that they finish simultaneously.  In the abstract language of queueing models, this amounts to adding blocking barriers to the parallel processing model.

In this section we mainly focus on machine learning applications, since that is driving a lot of current interest in parallel computation.  This should not be taken as a dismissal of the many other applications where 1 and 2-barrier parallelism arises.
%
%
\subsubsection{Barriers in Apache Spark}
Until recently, tools like Spark did not support this type of execution.  Project Hydrogen took as motivation the fact that Spark is a popular and powerful tool for data processing and preparation, and that it would be ideal to make high-performance ML tools available directly in Spark~\cite{project-hydrogen}.  To that end, beginning with Spark 2.4, Spark developers implemented support for Barrier Execution Mode RDDs~\cite{spark:bem-ticket}.  These are Resilient Distributed Datasets (RDDs), but with the constraint that the tasks executed on a BarrierRDD must start simultaneously (and any task failures mean the job must restart).  The API also includes support for placing intermediate barriers (synchronization points) in the tasks.

Horovod is a distributed training framework that integrates with many popular ML tools, and was originally implemented using MPI~\cite{horovod-arxiv,mpi-standard}.  Its use of a Ring-Allreduce necessitates start and departure barriers~\cite{patarasuk:ring-allreduce}.  A proprietary version of Horovod reportedly supports execution on Spark clusters and is implemented using BEM~\cite{spark:horovod-runner}, but has since been deprecated in favor of Torch Distributor~\cite{pytorch-ddp, torch-distributor}, which also uses Spark's barrier mode.  There is also an open-source version of Horovod that effects inter-task synchronization, including start and departure barriers on the tasks, but does not actually use the Spark BEM API ~\cite{horovod-on-spark, horovod-arxiv, spark:horovod-uber}.  Other examples of tools that use Barrier Execution Mode in Spark are , sparktorch~\cite{sparktorch-github}, and TensorFlowOnSpark~\cite{tensorflow-on-spark-github}.

Cannon's distributed matrix multiplication algorithm has also been implemented using Spark's Barrier Execution Mode~\cite{Foldi-2020}.  It also takes advantage of the OpenJDK native execution facilities introduced by Project Panama~\cite{openjdk:project-panama}, and the authors' experiments show it to be measurably faster than the implementation in SparkML.  Cannon's algorithm is grid-based, hence the need for start barriers, but the memory requirements are independent of the degree of parallelism.
%
%
\subsubsection{$(s,k,l)$ Barrier Systems}
Some machine learning practitioners have explored abandoning the synchronicity constraints during training in the interest of performance~\cite{hogwild-asyng-sgd,cohen-async-sgd}.  The disadvantage of this asynchronous parallel approach is that some workers may be training against a stale model state.  An intermediate strategy, which we call $(s,k,l)$, of retaining the barriers but preempting and discarding the slowest tasks of barrier workloads, can improve their stability and performance, especially in the 2-barrier case.  This is sometimes also referred to as running redundant tasks, or backup tasks.  This strategy is applicable in cases where the work being performed is to some extent fungible, which is the case with many machine learning workloads.  Discarding some tasks could potentially bias the effective training set, but it does not necessarily stop the model from making progress.  Several papers in the past decade have reported on the practical performance of this strategy, in terms of efficiency and training effectiveness~\cite{tail-at-scale, chen2017revisitingdistributedsynchronoussgd, xu-dynamic-backup-workers, skl-ddppo-iclr}.
%
%
\subsubsection{1 vs 2-Barrier Architectures}
Distributed ML systems usually follow two main architectures, which correspond roughly to 1 and 2-barrier BEM models.  In architectures that use an Allreduce, or are otherwise bulk synchronous, all workers must block until all (or enough) tasks finish, in order to update their collective state.  This is best modeled as a 2-barrier system.  In other architectures there are one or more parameter servers, where workers send their gradient updates when they finish their batch~\cite{scaling-ml-parameter-server}.  In these cases the worker is plausibly free to service tasks of another job, such as running inferences for validation or exploratory training.  This is more like a start barrier only, or 1-barrier system.
%
%
\subsubsection{Hybrid Barrier Systems}
The type of hybrid systems we study in section~\ref{sec:bem-performance} may have a mix of barrier and non-barrier jobs with varying numbers of tasks.  In ML for example inference jobs may come with varying batch sizes.  It also arises in the context of moldable schedulers that modulate the parallelism jobs are allocated at the time when they enter service\cite{cirne:using-moldability,optimal-server-allocation-moldable-jobs-concave,effective-selection-of-partiton-size-moldable}.  A number of contexts outside of parallel computing are described in~\cite{green1980-random-bem}.

We specifically look at the 1-barrier case.
Models similar or equivalent to the 1-barrier BEM architecture with random $k$ have been studied in~\cite{green1980-random-bem} which gives an expression for the Laplace transform of the waiting time for the M$\mid$M$\mid$$s$ case, and in~\cite{federgruen1984-bem-mgc-random} where it is approximated for the M$\mid$G$\mid$$s$ case.  However this only resulted in explicit solutions in the 2-barrier case for $s=2$ workers~\cite{brill-green-1984}.
%
%
\section{Stability in Barrier Execution Mode}
\label{sec:bem-stability}
A system is defined to be stable, if certain stationary/long-term performance measures exist.  I.e., if the backlog or the delay converges in distribution to the distribution of a finite random variable~\cite{chang:dynamicserviceguarantees,luebben:availbw-ton}. If the queueing system is modeled as a Markov chain, a necessary condition for stability is that the underlying Markov chain is recurrent non-null, i.e., that the system regularly empties within finite time~\cite{bolch:queueingnetworks}. This is true when the mean arrival rate is smaller than the mean service rate, i.e. $\lambda < \mu$, or equivalently when the utilization $\varrho = \lambda/\mu < 1$.  For parallel systems with $s$ workers the service rate of one worker is multiplied by $s$, i.e., $\lambda < s\mu$~\cite{bolch:queueingnetworks}. In the following we derive stability results for 1-barrier and 2-barrier systems with or without task redundancy. The results are summarized in Tab.~\ref{tab:stability}.
\begin{table}
\caption{Summary of stability conditions}
\label{tab:stability}
\begin{center}
\begin{tabular}{|c||c|c|} \hline
$s$ servers & 1-barrier jobs & 2-barrier jobs \\
\hline\hline
&&\\[0.01mm]
\makecell[r]{$k$ tasks per job} & \makecell{$\varrho < \frac{1}{\sum_{j=0}^{k-1} \frac{1}{k-j\frac{k}{s}}}$} & \makecell{$\varrho < \frac{1}{\sum_{j=0}^{k-1} \frac{1}{k-j}}$} \\
&&\\[0.01mm]
\hline
&&\\[0.01mm]
\makecell[r]{$k$ tasks per job \\ $l$ must complete \\ (total work)} & \makecell{no simple closed form} &
\makecell{$\varrho < \frac{1}{\sum_{j=0}^{l-1} \frac{1}{l-j\frac{l}{k}}}$} \\
&&\\[0.01mm]
\hline
&&\\[0.01mm]
\makecell[r]{$k$ tasks per job  \\ $l$ must complete \\ (useful work)} & \makecell{no simple closed form} &
\makecell{$\varrho < \frac{1}{\sum_{j=0}^{l-1} \frac{1}{l-j\frac{l}{k}}} - \frac{k-l}{k}$} \\
&&\\[0.01mm]
\hline
\end{tabular}
\end{center}
\vspace*{-3mm}
\end{table}
%
%
\subsection{Stability with Start and Departure Barriers}
\label{sec:bem-2-stability}
In a 2-barrier system (blocking start and departure barriers), if $s$ is an integer multiple of $k$, as in Fig.~\ref{fig:bem-scheduling}, all of the workers can be used.  If not, let $r = s \bmod{k}$.  Then there will always be at least $r$ workers idle, as the jobs' tasks must start and depart in groups of $k$.  For example, if $s=7$ and $k=4$ only one job can execute at a time, leaving three workers permanently idle.  Thus, a BEM system with 2 barriers and exponential inter-arrival times is equivalent to an M$\mid$G${\mid}m$ queue with $m=\lfloor s/k \rfloor$ and its service time distribution is the $k$th order statistic $X_{(k)}$ of the service time distribution of $k$ iid tasks (i.e. the maximum).

\begin{figure}
  \centering
  \includegraphics[width=0.95\columnwidth]{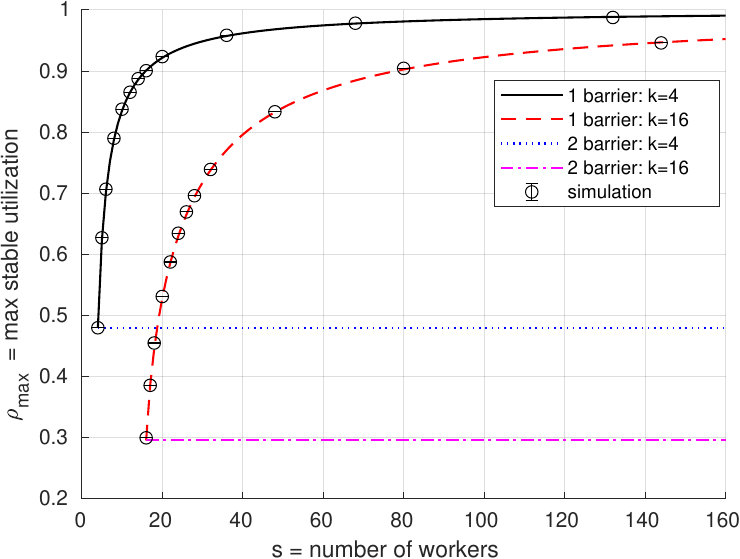}
  \caption{Maximum stable utilization for both 1 and 2-barrier BEM systems for varying numbers of workers.  Lines are analytical results from Eqns.~\eqref{eq:2barrier-stability} and~\eqref{eq:1barrier-stability-bound}.  Individual data points are from corresponding simulations~\cite{forkulator}.}
  \label{fig:utilization-stability}
  \vspace*{-3mm}
\end{figure}

For the following analysis we consider exponential task service times, and assume $s$ is an integer multiple of $k$.  Rewriting the maximum of $k$ exponential random variables with parameter $\mu$ as a sum, we have the minimum of $k$ exponentials as the time to finish the first task, then an additional time until the second task is finished, that is (due to memorylessness) the minimum of $k-1$ exponentials and so on until we obtain Renyi's formula
\begin{equation}
    X_{(k)} =_d \sum_{j=1}^{k} Y_j ,
    \label{eq:renyi-kofk}
\end{equation}
where $Y_j$ are independent exponential random variables with parameter $j\mu$~\cite[(1.9)]{renyi1953theory}. It follows that $X_{(k)}$ has expected value $\mathsf{E}[X_{(k)}] = \sum_{j=1}^k (j\mu)^{-1} = H_k/\mu$, where $H_k$ is the $k$th harmonic number. The sequence $(H_k)_{k\in\mathbb{N}}$ has the asymptotic limit $\gamma + \ln k$, where $\gamma \approx 0.577$ is the Euler constant.  Hence, the system is stable as long as
\begin{equation}
    \lambda < \frac{s}{k} \frac{\mu}{H_k} .
    \label{eq:2barrier-stability-lambda}
\end{equation}
The utilization of the system is $\varrho = k\lambda/(s\mu)$, therefore the maximum stable utilization of a 2-barrier system is given by
\begin{equation}
    \varrho < \frac{1}{H_k} = \frac{1}{\sum_{j=0}^{k-1} \frac{1}{k-j}}.
    \label{eq:2barrier-stability}
\end{equation}

Note that for a given $\varrho$ this depends on the number of tasks per job, $k$, but not on the number of workers, $s$.  This is apparent in Fig.~\ref{fig:utilization-stability}, where, for each $k$, the stability region is constant across all values of $s$.
%
%
\subsection{Stability with Start Barriers}
\label{sec:bem-1-stability}
Next we consider the same BEM system but with only a blocking start barrier.  If the system is continuously backlogged, whenever a job starts service, all workers are busy.  The time until $k$ workers become idle, so that the next job can start service, is the $k$th order statistic of $s$ exponential task service times (again using memorylessness). With the same reasoning as above (minimum of $s$ exponentials until the first worker is idle plus another minimum of $s-1$ exponentials until the next worker is idle and so on until $k$ workers are idle) we have Renyi's formula
\begin{equation}
    X_{(k)} =_d \sum_{j=s-k+1}^{s} Y_j,
    \label{eq:renyi-kofs}
\end{equation}
where $Y_j$ are exponential random variables with parameter $j \mu$. The expected value is
\begin{equation*}
    \mathsf{E}[X_{(k)}] = \frac{1}{\mu} \sum_{j=s-k+1}^{s} \frac{1}{j} = \frac{H_s-H_{s-k}}{\mu} ,
\end{equation*}
(note that $H_0=0$).  The rate at which the queue is served is the inverse of this expression and hence stability is given as long as
\begin{equation*}
    \lambda < \frac{\mu}{H_s - H_{s-k}} .
\end{equation*}
Defining the utilization of the system as $\varrho = k \lambda/(s\mu)$ as before, we have
\begin{equation}
    \varrho < \frac{k}{s(H_s-H_{s-k})} = \frac{1}{\sum_{j=0}^{k-1} \frac{1}{k-j\frac{k}{s}}} ,
    \label{eq:1barrier-stability-bound}
\end{equation}
for stability. We see an improvement of the stability region with declining parallelism ratio $k/s$ and we also see a general improvement over 2-barrier systems~\eqref{eq:2barrier-stability}, since $k/s \le 1$ (as noted, when $k=s$ they both become identical to Split-Merge). In the limit, if $s$ is large compared to $k$ we have
\begin{equation*}
    \lim_{k/s \rightarrow 0} \frac{1}{\sum_{j=0}^{k-1} \frac{1}{k-j\frac{k}{s}}} = 1 ,
\end{equation*}
so that $\varrho < 1$ for stability. We observe this asymptotic convergence to $1$ as $s\rightarrow \infty$ in Fig.~\ref{fig:utilization-stability}. For practical values of $k=10$ and $s=100$ a stable utilization of $0.95$ is possible.
%
%
\subsection{Stability of 2-Barrier $(s,k,l)$ Systems}
\label{sec:skl-2barrier-systems}
The analysis of the performance of two-barrier systems is well-understood, and the stability region turned out to be invariant with $s$ and rather poor.  And yet, two-barrier parallel jobs are important to machine learning and many other computational applications.  One practical optimization ML researchers have found and employed is that of using what we will call $(s,k,l)$ barrier jobs.  In an $(s,k,l)$ system there are $s$ workers, each job comprises $k \le s$ tasks, but the job departs as soon as $l \in [1,k] $ of the tasks complete.  The straggler jobs are killed, and the barrier blocking the job's workers is freed.  This is a reasonable approach to dealing with the variance in task duration in applications where the work being performed is fungible, as is the case in many ML applications~\cite{skl-ddppo-iclr}.

A conventional 2-barrier BEM system with exponential inter-arrival times is equivalent to an M$\mid$G$\mid$m system, where $m=\lfloor s/k \rfloor$, but in the $(s,k,l)$ case, the job service times have the distribution of the $l^{\text{th}}$ order statistic of $k$ iid exponentials.  Following the same reasoning as before, let $X_1, X_2, \ldots, X_k$ be iid exponentials with rate $\mu$.  The time until the first task departs, $X_{(1:k)}$ has the distribution of an exponential with rate $k\mu$.  Because of the memoryless property of the exponential, the time interval between the departure of the first and second task is the min of the remaining $k-1$ exponentials, $X_{(1:k-1)}$, which has the distribution of an exponential with rate $(k-1)\mu$.  Repeating this $l$ times, we see that
\begin{align}
    X_{(l:k)} =_d \sum_{j=0}^{l-1} Y_{k-j}
\end{align}
where $Y_j \sim \text{Exp}(j\mu)$.  Therefore
\begin{align}
    \mathsf{E}\left[ X_{(l:k)} \right] = \frac{1}{\mu} \sum_{j=0}^{l-1} \frac{1}{k-j} = \frac{H_k - H_{k-l}}{\mu}
\end{align}

Because all $k$ tasks of a 2-barrier job release their respective workers in unison when the $l^{\text{th}}$ task finishes, the system is stable as long as the arrival rate of $(s,k,l)$ jobs satisfies
\begin{align}
    \lambda < \frac{s}{k} \frac{1}{\mathsf{E}\left[ X_{(l:k)} \right]} = \frac{s\mu}{k (H_k - H_{k-l})} , \label{eq:skl-2barrier-lambda-max}
\end{align}
where we assume $k$ divides $s$.

The expected total server time used by an $(s,k,l)$ job turns out to have a simple form.  The total server time is the sum of the execution time of all the tasks.  There are two types of tasks: the first $l$ that complete, and the final $(k-l)$, which are terminated.  Let $J_{tot} = \sum_{i=1}^{l}X_{(i:k)}+(k-l)X_{(l:k)}$ be the total server time used by an $(s,k,l)$ job.  Then the expectation for such a job with exponential task service rate $\mu$ can be computed by reordering the sum of harmonic numbers to group like terms.
\begin{align}
    \mathsf{E}&\left[ J_{tot} \right] = \sum_{i=1}^{l} \mathsf{E} \left[ X_{(i:k)} \right] + (k-l)\mathsf{E}\left[ X_{(l:k)} \right] \\
    &= \frac{1}{\mu}\left(\sum_{i=1}^{l} \sum_{j=k-(i-1)}^{k} \frac{1}{j} + (k-l)\sum_{j=k-(l-1)}^{k} \frac{1}{j} \right)  \\
    &= \frac{1}{\mu}\left( \sum_{j=k-(l-1)}^{k} \sum_{i=1}^{j-(k-l)}\frac{1}{j} + (k-l)\sum_{j=k-(l-1)}^{k} \frac{1}{j} \right)  \\
    &= \frac{1}{\mu}\sum_{j=k-(l-1)}^{k} \left( (j-(k-l)) + (k-l) \right) \frac{1}{j} \\
    &= \frac{1}{\mu} \sum_{j=k-(l-1)}^{k} 1 = \frac{l}{\mu}
\end{align}

Computing the useful utilization in an $(s,k,l)$ system is complicated by the fact that the server time spent on the final $(k-l)$ tasks is discarded.  Let $J_{useful}$ be the server time consumed for non-discarded work in an $(s,k,l)$ job.  The easiest way to isolate the expectation of $J_{useful}$ is to subtract off the expectation of the discarded server time.
\begin{align}
    \mathsf{E}\left[ J_{useful} \right] &= \frac{l}{\mu} - (k-l)\mathsf{E}\left[ X_{(l:k)} \right] \\
    &= \frac{1}{\mu}\left( l - (k-l)\left( H_k - H_{k-l} \right) \right)
\end{align}

The utilization of the system is $\varrho_{skl} = \lambda E[J_{tot}] /s$. With the maximum stable arrival rate $\lambda$ in~\eqref{eq:skl-2barrier-lambda-max}
we can compute the maximum stable utilization
\begin{align}
    \varrho_{skl} < \frac{l}{k (H_k - H_{k-l})}
    \label{eq:skl-stability}
\end{align}
and the maximum useful utilization can be simplified as
\begin{align}
    \varrho_{useful} \, < \, & \frac{s \mu (l- (k-l)(H_k - H_{k-l}))}{k s \mu (H_k - H_{k-l})} \nonumber \\
    &= \frac{(l- (k-l)(H_k - H_{k-l}))}{k (H_k - H_{k-l})} \nonumber \\
    &= \frac{l}{k(H_k - H_{k-l})} - \frac{k-l}{k} = \varrho_{skl} - \frac{k-l}{k} \label{eq:skl-stability-useful}
\end{align}

Expanding the harmonic sums in~\eqref{eq:skl-stability} and~\eqref{eq:skl-stability-useful} gives the formulas shown in Table~\ref{tab:stability}.
We note that both $\varrho_{skl}$ and $\varrho_{useful}$ reduce to~\eqref{eq:2barrier-stability}, the stability region of conventional BEM systems, when $l=k$.

Using these bounds, we can compute under which, if any, values of $l$ does using $(s,k,l)$ jobs improve the expected stable load over conventional BEM.  The left side of Fig.~\ref{fig:stability_qbbskl} shows the max stable utilization for a 2-barrier $(s,k,l)$ system with $s=k=32$ and varying $l$, based on the analytical results above, and simulation.  We see that for $l$ slightly less than $k$, the maximum stable useful utilization of the $(s,k,l)$ system exceeds conventional BEM, but for values of $l$ less than $10$ (for $s=k=16$), the $(s,k,l)$ system performs worse.

The benefit of the $(s,k,l)$ strategy will of course depend on the distribution of the task service times.  Exponential task service times allow us to derive analytical results and gain intuition about the situation, but the memoryless nature of the task spacings limit the gain that can be observed here.  If one assumes, a bimodal exponential task service time distribution with $10\%$ of the tasks having rate $\mu/10$ or $\mu/100$, the advantage of using $(s,k,l)$ and preempting stragglers is much more pronounced.  This is mainly because the baseline $\varrho_{bem}$ stability region becomes so much worse, while $(s,k,l)$ scheduling preserves the stability region for most $l$.  This is shown in the right side of Fig.~\ref{fig:stability_qbbskl}, and we see that the system with an optimally chosen $l$ can support a utilization about three times higher than conventional BEM for these parameters.

\begin{figure}
  \centering
  \includegraphics[width=0.95\columnwidth]{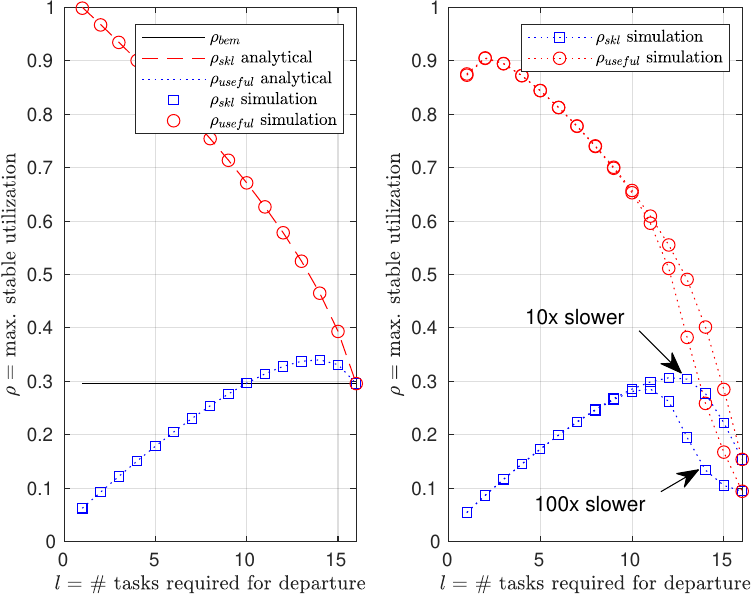}
  \caption{Maximum stable utilization for 2-barrier $(s,k,l)$ BEM system with $s=k=16$ and varying $l$, from the bounds computed in section~\ref{sec:skl-2barrier-systems}, and simulation.  The left plot is for iid exponential task service times.  The right plot is for a bimodal exponential service process, where there is a 10\% probability of receiving slower service.}
  \label{fig:stability_qbbskl}
  \vspace*{-2mm}
\end{figure}
%
%
\subsection{Stability of 1-Barrier $(s,k,l)$ Systems}
\label{sec:skl-1barrier}

The stability situation in the 1-barrier $(s,k,l)$ case is more complicated, since the number of active jobs, and the number of workers they occupy can change.  We will give an outline of how one can still compute the stability region of this type of system  with exponential inter-arrival times and exponential task service times using a continuous-time Markov chain (CTMC) model.

Each job running in a 1-barrier $(s,k,l)$ system has between $(k-l+1)$ and $k$ tasks running.  We represent the state of such a system as a vector of length $l$, $\vec{S}=(c_{k-l+1}, \ldots, c_k)$ where $c_r$ is the number of active jobs with $r$ running tasks.  For any state $\vec{S}$, the number of running tasks is
\begin{align}
    T(\vec{S})=\sum_{r=k-l+1}^k r\cdot c_r .
\end{align}

For example, for $(s=12, k=4, l=3)$, each job can have either $2$, $3$, or $4$ tasks running. The state $S=(1,0,2)$ has one active job with $2$ tasks running and two jobs with $4$ tasks running, so a total of $T(S) = 1\cdot 2 + 2\cdot 4 = 10$ of the $12$ workers are busy.

For the purposes of computing the stability region, we are only interested in states where the system is backlogged.  That is, the number of idle workers is less than $k$, and new jobs are blocked from starting.

Let $\mathcal{S}_i$ for $i=(s-k+1), \ldots, s$ be the set of feasible states with $i$ busy workers:
\begin{align}
    \mathcal{S}_i = \left\{ \vec{S}=(c_{k-l+1}, \ldots, c_k) \mid T(\vec{S})=i \right\}
\end{align}

Then computing the number of states in $\mathcal{S}_i$ is equivalent to computing a restricted version of the integer partition function of $i$ such that each part is in the range $(l+1), \ldots, k$~\cite{andrews2004-integer-partitions}.  Then the total set of feasible states is the union of these $\mathcal{S}_i$
\begin{align}
    \mathcal{S}=\bigcup_{i=s-k+1}^{s}\mathcal{S}_i
\end{align}
The partition function does not have a closed form expression, but it can be computed by recursion or dynamic programming, and the states of the CTMC can be enumerated in this way as well.  Just as with partitions, the states have a partial ordering, but no total ordering.

The actual number of states in the CTMC is sometimes slightly less, because some feasible states are unreachable.  For example, if $(s,k,l)=(15, 5, 3)$, then the state $(5,0,0)$ is technically feasible, but unreachable, because from state $(4,0,0)$ no new jobs can start.  However the state $(0,\ldots, 0, \lfloor s/k \rfloor)$ is always recurrent, and can be used as a starting point for enumerating the other states.

We build a CTMC by enumerating the states, and adding four types of transitions.  The specific transitions possible from a state depend on the parallelism of the active jobs it contains.  Given a state $\vec{S}=(c_{k-l+1}, \ldots, c_k)$, for each $c_r>0$, $\vec{S}$ has one of the following out-transitions with rate $c_r \mu$.

\begin{enumerate}[(i)]
    \item $r>k-l+1$ and $T(\vec{S})>s-k+1$ : \\
        $c_r\rightarrow c_r-1$, $c_{r-1}\rightarrow c_{r-1}+1$ \\
        no new job starts
    \item $r>k-l+1$ and $T(\vec{S})=s-k+1$ : \\
        $c_r\rightarrow c_r-1$, $c_{r-1}\rightarrow c_{r-1}+1$, $c_k\rightarrow c_k+1$ \\
        new job starts
    \item $r=k-l+1$ and $T(\vec{S})>s-l+1$ : \\
        $c_r\rightarrow c_r-1$ \\
        no new job starts
    \item $r=k-l+1$ and $T(\vec{S})\leq s-l+1$ : \\
        $c_r\rightarrow c_r-1$, $c_k\rightarrow c_k+1$ \\
        new job starts
\end{enumerate}

Solving this CTMC gives us a steady-state distribution on the states of a 1-barrier $(s,k,l)$ BEM system in a perpetually backlogged state.  We can use this to compute the maximum sustained job throughput of such a system as a weighted sum of the rates of the associated type (ii) and type (iv) transitions.  This gives the expected rate at which new jobs can start.  We know from the previous section that the expected total server time taken by each job is equal to $\frac{l}{\mu}$, and the expected useful server time taken per job is $\frac{1}{\mu}\left( l - (k-l)\left( H_k - H_{k-l} \right) \right)$.  This allows us to compute the maximum stable utilization, and maximum stable useful utilization levels for any particular values of $(s,k,l)$.

Fig.~\ref{fig:stability_qbskl} shows the maximum stable utilization for 1-barrier $(s,k,l)$ BEM systems for $k=16$, $l=8$, and varying $s$, along with the maximum stable utilization of conventional 1-barrier BEM systems based on solutions of the CTMC and direct simulation.  For comparison it also shows the maximum stable utilization of a conventional BEM system.  We see that the total utilization of the 1-barrier $(s,k,l)$ system fluctuates quite a bit for small $s$, but converges to that of the conventional BEM system as $s$ increases. Also the useful utilization of the $(s,k,l)$ system is much lower.  In 1-barrier systems, workers are not blocked after short-lived tasks depart, so killing straggler tasks does not free any workers that would otherwise be idle.  The $(s,k,l)$ strategy may speed up job departures, but it is not beneficial to the useful system utilization.

\begin{figure}
  \centering
  \includegraphics[width=0.86\columnwidth]{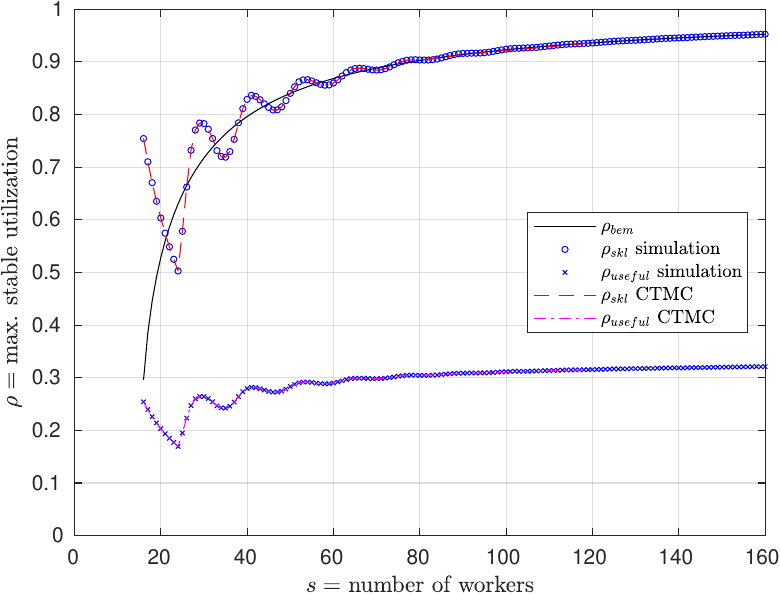}
  \caption{Maximum stable utilization for 1-barrier $(s,k,l)$ BEM system with $k=16$, $l=8$, and varying $s$, from the solution of the CTMC described in section~\ref{sec:skl-1barrier}, and simulation.}
  \label{fig:stability_qbskl}
\end{figure}
%
%
\subsection{Performance with Respect to Task Throughput}
\label{sec:skl-task-throughput}
In the machine learning context, the tasks typically originate from splitting the training data up into batches.  Some may take longer to process than others, but the value of each completed task may be regarded as the same, regardless of how much server time it consumed.  This seems to be the philosophy in~\cite{skl-ddppo-iclr}, for example.  This is a very different metric from traditional utilization, where the throughput of the system is measured in terms of maximizing the fraction of time the workers are busy, and the value of a task is proportional to the server time that it consumes.

We can adjust the analysis above to evaluate the maximum stable task throughput relative to the number of workers.
For the 2-barrier case we use~\eqref{eq:skl-2barrier-lambda-max} to compute
\begin{align}
    \widetilde{\varrho}_{useful} = \frac{s\mu}{k(H_k - H_{k-l})} \frac{l}{s\mu} = \frac{l}{k(H_k - H_{k-l})}
\end{align}
But this is exactly $\varrho_{skl}$ for the 2-barrier case, and it has a maximum utilization of 1 for $l=1$, as seen in Fig.~\ref{fig:stability_qbbskl}.  The only purpose of the parallelism is to quickly determine which task wins the race.  This is an artifact of the properties of min order statistics of exponentials.  We conclude that exponential task service times are not a good model for evaluating BEM systems with regard to the metric of pure task throughput.
%
%
\section{Performance Bounds for BEM Systems}
\label{sec:bem-performance}

In the last decade, stochastic network calculus has been used to study queueing in parallel systems such as Fork-Join, Split-Merge, and many  other variants~\cite{kesidis:forkjoin,rizk:forkjoin,ciucu-replication-systems,rizk-parallel-systems,fidler-tpds,fidler-tiny-tasks,rizk:output-synchrnonization}.  We will extend these results to derive bounds for 1-barrier BEM systems with exponential task service times. Approximations for the mean waiting and sojourn times of 2-barrier BEM systems follow readily from queueing theory, and we will use them for comparison.
%
%
\subsection{Max-Plus System Model}
\label{sec:model}
We will build on the results of~\cite{fidler-tpds} which use max-plus system theory to derive performance bounds for a variety of parallel systems.  First we need to introduce some notation and definitions.

Jobs are labeled in order of arrival with {\bf arrival time} $A(n)$ and {\bf departure time} $D(n)$, for $n \ge 1$.  We denote job inter-arrival times by $A(m,n) = A(n) - A(m)$ for $n \ge m \ge 1$.  The {\bf sojourn time} is defined as $T(n) = D(n) - A(n)$.  We use the following definition from~\cite[Def. 1]{fidler-tpds}, which in turn is adapted from ~\cite[Def. 6.3.1]{chang:performanceguarantees}.  It uses the idea of a {\it service process}, $S(m,n)$, that is related, but not generally identical, to the cumulated service time of jobs $m$ to $n$.
\begin{definition}[Max-plus server]
\label{def:maxplusserviceprocess}
A system with arrivals $A(n)$ and departures $D(n)$ is an $S(m,n)$ server under the max-plus algebra if it holds for all $n \ge 1$ that
\begin{equation*}
D(n) \le \max_{m \in [1,n]} \{ A(m) + S(m,n) \} .
\end{equation*}
It is an exact $S(m,n)$ server if it holds for all $n \ge 1$ that
\begin{equation*}
D(n) = \max_{m \in [1,n]} \{ A(m) + S(m,n) \} .
\end{equation*}
\end{definition}

It follows by insertion of Def.~\ref{def:maxplusserviceprocess} that for an $S(m,n)$ server
\begin{equation*}
T(n) \le \max_{m \in [1,n]} \{ S(m,n) - A(m,n) \} .
\end{equation*}

The function $S(m,n)$ can be thought of as the amount of service provided by the system between jobs $m$ and $n$.  For a simple example, consider a single work-conserving worker servicing jobs with constant rate $\mu$.  If we take
\begin{equation*}
S(m,n) = \frac{n-m+1}{\mu}
\end{equation*}
then the system is an exact $S(m,n)$ server.  For any particular $n$, the $m$ selected by the maximum will be the start of the most recent busy period, or $m=n$ if the system is idle when job $n$ arrives at time $A(n)$. Note that $S(n,n)=\frac{1}{\mu}$ is the time required to service one job.

In order to state network calculus performance bounds, we need to define $(\sigma,\rho)$ arrival and service envelopes~\cite[Def. 7.2.1]{chang:performanceguarantees}.  The parameters $\sigma$ and $\rho$ specify an affine bounding function, i.e., intercept and slope, of the logarithm of the moment generating function (MGF), and can be thought of as a stochastic version of the burst and rate parameters of a leaky-bucket regulator. The MGF of a random variable $X$ is defined as $\mathcal{M}_X(\theta) =\mathsf{E}\bigl[e^{\theta X}\bigr]$ where $\theta$ is a free parameter.  The following definition is the same as~\cite[Def. 2]{fidler-tpds}.
\begin{definition}[$(\sigma,\rho)$-Arrival and Service Envelopes]
\label{def:sigmarho}
An arrival process is $(\sigma_A,\rho_A)$-lower constrained if for all $n \ge m \ge 1$ and $\theta > 0$ it holds that
\begin{equation*}
\mathsf{E}\Bigl[e^{-\theta A(m,n)}\Bigr] \le e^{-\theta (\rho_A(-\theta) (n-m) - \sigma_A(-\theta))} .
\end{equation*}
Similarly, a service process is $(\sigma_S,\rho_S)$-upper constrained if for all $n \ge m \ge 1$ and $\theta > 0$ it holds that
\begin{equation*}
\mathsf{E}\Bigl[e^{\theta S(m,n)}\Bigr] \le e^{\theta(\rho_S(\theta) (n-m+1) + \sigma_S(\theta))} .
\end{equation*}
\end{definition}

In the special case of general arrival processes with independent increments (GI arrivals), $A(m,n) = \sum_{\nu=m}^{n-1} A(\nu,\nu+1)$ has iid inter-arrival times $A(\nu,\nu+1)$. It follows that $\mathsf{E}\bigl[e^{-\theta A(\nu,\nu+1)}\bigr] = \mathsf{E}\bigl[e^{-\theta A(1,2)}\bigr]$ for $\nu \ge 1$. Next, we use that the MGF of a sum of independent random variables is the product of their individual MGFs, i.e.,
$\mathsf{E}\bigl[e^{-\theta A(m,n)}\bigr] = \mathsf{E}\bigl[e^{-\theta A(1,2)}\bigr]^{n-m}$
to derive minimal traffic parameters from Def.~\ref{def:sigmarho} as $\sigma_A(-\theta) = 0$ and
\begin{equation*}
\rho_A(-\theta) = -\frac{1}{\theta} \ln \mathsf{E}\Bigl[e^{-\theta A(1,2)}\Bigr].
\end{equation*}
Similarly for GI service processes, $S(m,n)$ is composed of iid service increments $S(\nu,\nu)$, i.e., $S(m,n) = \sum_{\nu=m}^n S(\nu,\nu)$, so that it has minimal parameters $\sigma_S(\theta) = 0$ and
\begin{equation*}
\rho_S(\theta) = \frac{1}{\theta} \ln \mathsf{E}\Bigl[e^{\theta S(1,1)}\Bigr].
\end{equation*}
Parameter $\rho_A(-\theta)$ decreases with $\theta > 0$ from the mean to the minimum inter-arrival time and $\rho_S(\theta)$ increases with $\theta > 0$ from the mean to the maximum service time.  The two definitions above are enough to prove~\cite[Thm. 1]{fidler-tpds}, which provides the primary bounds on waiting and sojourn time distributions that we are building on.
%
%
\subsection{Performance Bounds for 1-Barrier BEM Systems}
\label{sec:modelbarrierexecutionmode}
The analysis of this model is a generalization of the results for purely BEM systems in~\cite{bem-infocom}, and makes use of the approach in~\cite{fidler-tpds} for Single-Queue Load Balancing systems.  We want to study the performance of hybrid systems that service a mix of BEM and non-BEM jobs with a heterogeneous degrees of parallelism.  We introduce some more flexible notation for expressing the spacings between start times of the jobs.  Here we use the start time of each job's last task as the reference point to quantify these inter-job time intervals.

Let $V_{n,i}$ be the time when task $i$ of job $n$ starts service.  Let $X_j(V_{n,i})$ be the residual service time of the task at worker $j$ at time $V_{n,i}$, or zero if the worker is idle.  Finally, for $l\in [1,s]$, let $Z_l(V_{n,i}) = \{ X_1(V_{n,i}), \dots, X_s(V_{n,i})\}_{(l)}$ be the $l$th order statistic of the residual service times at all the workers at time $V_{n,i}$.  I.e. $Z_1$ would give us the waiting time until the next worker becomes free.

This gives us the machinery we need to express the spacing between jobs in a hybrid system.  In the lemma below, $\Omega(n)$ will be the time between the start of the final task of job $n$, and the start of the final task of job $n+1$.

\begin{lemma}[Generalized Barrier Execution Mode system]
\label{lem:barrierexecutionmode}
Let $K(n)$ be a sequence of discrete random variables with support $[1,s]$.  Consider a system with $s$ parallel work-conserving workers, servicing a mix of jobs with and without a blocking start barrier, where job $n$ comprises $K(n)$ tasks.
Let $Q_i(n)$ denote the service time of task $i \in [1,K(n)]$ of job $n$.  Let $V_{n,i}$, $X_j(V_{n,i})$, $Z_l(V_{n,i})$, and $\Omega(n)$ be as defined above.

For $n \ge m \ge 1$ define
\begin{equation*}
S^G(m,n) = \max_{i \in [1,K(n)]} \Bigl\{ Q_i(n) \Bigr\} + \sum_{\nu=m}^{n-1} \Omega(\nu) ,
\end{equation*}
where
\begin{equation*}
\Omega(n) =
\begin{cases}
Z_{K(n+1)}(V_{n,K(n)}), \qquad \text{job } n+1 \text{ has start barrier,} \\
Z_1(V_{n,K(n)}) + \sum_{i=1}^{K(n+1)-1} Z_1(V_{n+1,i}), \hfill \text{otherwise.}
\end{cases}
\end{equation*}

Then the system is an $S^{G}(m,n)$ server.
\end{lemma}
\begin{proof}
The system is initially idle, so $V_{1,i}=A(1)$ for all $i \in [1,K(1)]$ for BEM and non-BEM jobs likewise.

Recall if job $n$ is a BEM job then $V_{n,i}=V_{n,i^\prime}$ $\forall i, i^\prime\in[1,K(n)]$, so we will write $V_{n,\ast}$ in these cases to emphasize that.

If job $2$ is a BEM job, then
\begin{equation}
    V_{2,\ast} = \max \{ A(2), V_{1,K(1)} + Z_{K(2)}(V_{1,K(1)}) \}.
    \label{eq:bem-vn-base}
\end{equation}
If job $2$ is a non-BEM job
\begin{equation*}
    V_{2,i} =
\begin{cases}
    \max \{ A(2), V_{1,K(1)} + Z_1(V_{1,K(1)}) \},  & \text{for } i=1 \\
    \max \{ A(2), V_{2,i-1} + Z_1(V_{2,i-1}) \} , & i \in [2,K(2)].
\end{cases}
\end{equation*}
For $n \ge 2$, if job $n$ is a BEM job
\begin{equation*}
    V_{n,\ast} = \max \{ A(n), V_{n-1,K(n-1)} + Z_{K(n)}(V_{n-1,K(n-1)}) \} ,
\end{equation*}
and if job $n$ is a non-BEM job
\begin{align}
    V_{n,K(n)} = \max \{ A(n), V_{n-1,K(n-1)} + &Z_1(V_{n-1,K(n-1)}) \nonumber\\
                                 + &\sum_{i=1}^{K(n)-1} Z_1(V_{n,i}) \} . \nonumber
\end{align}
To combine both cases into one recurrence relation, we write
\begin{equation}
    V_{n,K(n)} = \max \{ A(n), V_{n-1,K(n-1)} + \Omega(n-1) \} ,
\end{equation}
where $\Omega(n-1) = Z_{K(n)}(V_{n-1,K(n-1)})$ if job $n$ is a BEM job and $\Omega(n-1) = Z_1(V_{n-1,K(n-1)}) + \sum_{i=1}^{K(n)-1} Z_1(V_{n,i})$ if job $n$ is a non-BEM job. By recursive substitution we obtain
\begin{equation}
    V_{n,K(n)} = \max_{m\in [1,n]} \left\{ A(m) + \sum_{\nu=m}^{n-1} \Omega(\nu) \right\} .
    \label{eq:bem-vn-recursion-solution}
\end{equation}

\begin{figure}
  \centering
  \includegraphics[width=0.9\columnwidth]{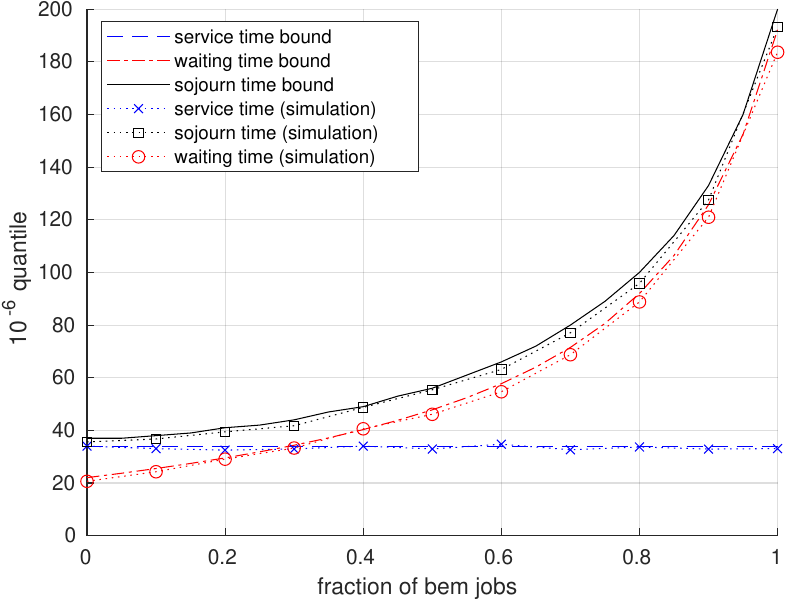}
  \caption{The $10^{-6}$ quantile of service time, waiting time, and sojourn time for a system servicing a hybrid mix of BEM and non-BEM jobs, with utilization $\rho=0.7$.  The $x$-axis spans the range from no BEM jobs to all BEM jobs.}
  \label{fig:bem-hybrid}
  \vspace*{-2mm}
\end{figure}
The departure time for task $i$ of job $n$ is its start time plus its service time
\begin{equation}
    D_i(n) = V_{n,i} + Q_i(n) ,
    \label{eq:bem-vn-departure}
\end{equation}
and the departure time for the entire job $n$ is $D(n) = \max_{i\in[1,K(n)]} \{ D_i(n) \}$. By insertion
\begin{equation*}
    D(n) = \max_{i\in[1,K(n)]} \{ V_{n,i} + Q_i(n) \} \le V_{n,K(n)} + Q(n) ,
\end{equation*}
where $Q(n) = \max_{i\in[1,K(n)]}\{ Q_i(n) \}$ is the maximum task service time.
By substitution of~\eqref{eq:bem-vn-recursion-solution} we have
\begin{align}
    D(n) &\le \max_{m\in [1,n]} \left\{ A(m) + Q(n) + \sum_{\nu=m}^{n-1} \Omega(\nu) \right\} \label{eq:bem-recursion-substitution} ,
\end{align}
and with $S^{G}(m,n) = Q(n) + \sum_{\nu=m}^{n-1} \Omega(\nu)$ the BEM model is an $S^{G}(m,n)$ server.
\end{proof}

\begin{figure}
  \centering
  \includegraphics[width=0.9\columnwidth]{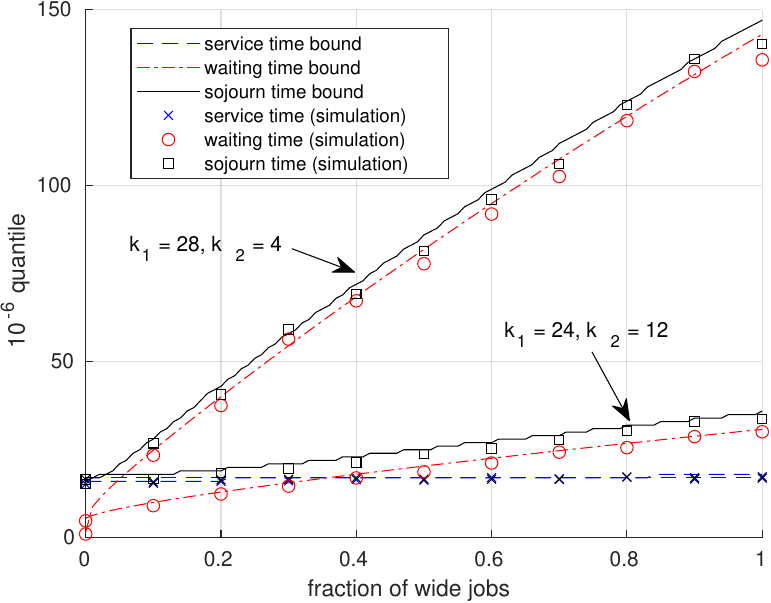}
  \caption{The $10^{-6}$ quantile of execution, waiting, and sojourn times for a system servicing BEM jobs with utilization $\rho=0.4$ a bimodal task size distribution (jobs may have $k_1$ or $k_2$ tasks).  The $x$-axis spans the range from all skinny to all wide jobs.}
  \label{fig:bem-hybrid-k}
  \vspace*{-2mm}
\end{figure}

Lemma~\ref{lem:barrierexecutionmode} holds for general task service times.  In order to obtain concrete bounds, we need to assume a bit more. The difficulty in using the above result is that we need to determine the residual service times $X_j(V_{n,i})$, some of which may be zero, which complicates the behavior of the order statistics. However, if we assume iid exponential task service times, $Q_i(n)$, and replace the zero entries in $X_j(V_{n,i})$ with identical, independent exponential random variables, we obtain a service process that does produce a computable bound.

\begin{corollary}[]
Consider a system  as in Lemma~\ref{lem:barrierexecutionmode} and assume the task service times, $Q_i(n)$, are iid exponential with parameter $\mu$.  Let $\hat{Z}_l$ be the $l$th order statistic of $s$ iid exponentials with parameter $\mu$, and define $\Omega(n)$ and $S(m,n)$ as before, using these $\hat{Z}_l$.
Then the system is an $S(m,n)$ server.
\end{corollary}

\begin{proof}
Here we assume the task service times are iid exponential with parameter $\mu$. At any time $V_{n,i}$ some workers are idle and some are servicing tasks.  By the memorylessness of the exponential distribution, the residual service times of busy workers are also exponential with parameter $\mu$.
Thus, if we replace the zero elements of $\{X_j(V_{n,i})\}_{j\in[1,s]}$ (those corresponding to idle workers), we have a set of $s$ iid exponentials with parameter $\mu$.  Then $\hat{Z}_l(V_{n,i})$ is the $l$th order statistic of this modified set.  Because exponential samples are strictly positive, $\hat{Z}_l(V_{n,i}) \ge Z_l(V_{n,i})$.  Therefore $S(m,n) \ge S^G(m,n)$ and \eqref{eq:bem-recursion-substitution} holds also for $S(m,n)$, and the system is an $S(m,n)$ server.
\end{proof}

In order to derive computable performance bounds, we first need to express $\mathcal{M}_{S(m,n)}$, the MGF of $S(m,n)$,  in a form compatible with Def.~\ref{def:sigmarho}.  The results above combine two types of heterogeneity: a mix of BEM/non-BEM jobs, and a mix of jobs with random degrees of parallelism.  In order to limit the complexity of this step, we will carry out computations, and present results, for these two cases separately.  It should be clear how they could be combined.\\

\noindent {\bf BEM/non-BEM case:}\\
For this case we assume that $K(n)=k$ is fixed for all $n\ge1$.

First recall that for an exponential random variable $X$ with parameter $\mu$, the MGF is $\mathcal{M}_X(\theta) = \mu/(\mu-\theta)$, for $\theta < \mu$.  $Q(n)$ is a max order statistic of $k$ such exponentials, and we can again use~\eqref{eq:renyi-kofk} and \eqref{eq:renyi-kofs} to express order statistics as a sum.  Then the property of MGFs, that $\mathcal{M}_{X+Y}(\theta) = \mathcal{M}_X(\theta) \mathcal{M}_Y(\theta)$ for $X$ and $Y$ independent converts this into a product, giving
\begin{equation}
    \label{eq:bem-mgfQ}
    \mathcal{M}_{Q(n)}(\theta) = \prod_{j=1}^{k} \frac{j\mu}{j\mu - \theta} ,
\end{equation}
for $\theta < \mu$.  This gives us the first half of $\mathcal{M}_{S(m,n)}$.

Next we need to derive  the MGF of $\Omega(n)$.
Because the $\hat{Z}_l$ are order statistics of $s$ iid exponentials with parameter $\mu$, by the same reasoning we have
\begin{equation}
\mathsf{E}[e^{\theta \Omega(n)}|\text{job } n+1 \text{ is BEM}] = \prod_{j=s-k+1}^{s} \frac{j\mu}{j\mu - \theta} ,
\end{equation}
for $\theta < (s-k+1)\mu$.  Since the minimum of $s$ iid exponential random variables with parameter $\mu$ is exponential with parameter $s\mu$
\begin{equation}
\mathsf{E}[e^{\theta \Omega(n)}|\text{job } n+1 \text{ is non-BEM}] = \biggl( \frac{s\mu}{s\mu - \theta} \biggr)^k,
\end{equation}
for $\theta < s\mu$. Assuming jobs are iid BEM/non-BEM with probability $p_{\text{BEM}}$ and $p_{\text{non-BEM}} = 1-p_{\text{BEM}}$ it follows that
\begin{equation*}
\mathcal{M}_{\Omega(n)} (\theta) = p_{\text{non-BEM}} \biggl( \frac{s\mu}{s\mu - \theta} \biggr)^k + p_{\text{BEM}} \prod_{j=s-k+1}^{s} \frac{j\mu}{j\mu - \theta} ,
\end{equation*}
for $\theta < (s-k+1)\mu$.\\

\noindent {\bf Heterogeneous $K(n)$ case:}\\
For this case we assume that all jobs are BEM jobs, and therefore $\Omega(n) = \hat{Z}_{K(n+1)}(V_{n,K(n)})$.  Recall that the $\hat{Z}_{K(n)}$ are $K(n)$th order statistics of $s$ iid exponentials, but because the $K(n)$ are random variables, we need to uncondition. Assuming $K(n)$ are iid with probability $P(K(n) = k) = p(k)$
\begin{align}
\mathcal{M}_{\Omega(n)} (\theta) &= \mathsf{E} \Biggl[ \prod_{j=s-K(n+1)+1}^{s} \frac{j\mu}{j\mu - \theta} \Biggr] \nonumber \\ = & \sum_{k=1}^s p(k) \prod_{j=s-k+1}^{s} \frac{j\mu}{j\mu - \theta} ,
\end{align}
for $\theta < (s-k_{\max}+1)\mu$, where $k_{\max} = \max_n\{K(n)\}$.

The MGF of $Q(n)$ in this case also requires unconditioning.  We have
\begin{equation}
\mathsf{E} [ e^{\theta Q(n)} | K(n) = k ] = \prod_{j=1}^{k} \frac{j\mu}{j\mu - \theta},
\end{equation}
for $\theta < \mu$, and therefore
\begin{equation}
\mathcal{M}_{Q(n)}(\theta) = \mathsf{E} \Biggl[ \prod_{j=1}^{K(n)} \frac{j\mu}{j\mu - \theta} \Biggr] = \sum_{k=1}^s p(k) \prod_{j=1}^{k} \frac{j\mu}{j\mu - \theta} .
\end{equation}

Now, in both cases, taking
\begin{align}
    \sigma_S(\theta) &= \frac{1}{\theta} \ln \left[ \mathcal{M}_{Q(n)}(\theta) \right] - \frac{1}{\theta} \ln \left[ \mathcal{M}_{\Omega(n)}(\theta) \right] , \nonumber \\
    \rho_S(\theta) &= \frac{1}{\theta} \ln \left[ \mathcal{M}_{\Omega(n)}(\theta) \right] \label{eq:bem-rhoS} ,
\end{align}
we can write an expression for the MGF of $S(m,n)$ that satisfies Def.~\ref{def:sigmarho}, i.e.,
$\mathcal{M}_{S(m,n)}(\theta) = e^{\theta(\rho_S(\theta)(n-m+1) + \sigma_S(\theta))}$.
This means that we can apply basic results of the stochastic network calculus to derive performance bounds. A technicality that is due to the different ranges of the free parameter $\theta$ in $\mathcal{M}_{Q(n)}(\theta)$ and $\mathcal{M}_{\Omega(n)}(\theta)$ applies when optimizing $\theta$. For this reason, it is favorable with respect to tightness to first derive a bound on the waiting time, which only depends on $\mathcal{M}_{\Omega(n)}(\theta)$.
%
%
\subsection{Evaluating the Performance Bounds}
In BEM systems all tasks start simultaneously, but in a heterogeneous system they may not.  In order to avoid ambiguity, we define {\bf waiting time} of job $n$
\begin{equation*}
    W(n) = V_{n,K(n)} - A(n)
\end{equation*}
to be the time from its arrival, until the time when its last task starts service. $V_{n,K(n)}$ is given by~\eqref{eq:bem-vn-recursion-solution}.

The following theorem is essentially the same as Thm.~3 of~\cite{fidler-tpds}, but the mechanics of evaluating it will be different.

\begin{theorem}[Barrier Execution Mode system]
\label{th:barrierexecutionmode}
Consider a BEM system as in Lem.~\ref{lem:barrierexecutionmode}, with arrival parameters $(\sigma_A(-\theta),\rho_A(-\theta))$ satisfying Def.~\ref{def:sigmarho}, iid exponential task service times with parameter $\mu$, and parameter $\rho_S(\theta)$ given by~\eqref{eq:bem-rhoS}. For $n \ge 1$, the waiting time satisfies
\begin{equation}
\mathsf{P}[W(n) > \tau] \le \alpha e^{-\theta\tau} .
\label{eq:bem-waitingtime-bound}
\end{equation}
For G$\mid$M arrival and service processes, the free parameter $0< \theta < (s-k_{\max}+1) \mu$ has to satisfy $\rho_{S}(\theta) < \rho_A(-\theta)$ and
\begin{equation*}
\alpha = \frac{e^{\theta \sigma_A(-\theta)}}{1-e^{-\theta (\rho_A(-\theta) - \rho_S(\theta))}} .
\end{equation*}
For GI$\mid$M arrival and service processes, $0 < \theta < (s-k_{\max}+1) \mu$ has to satisfy $\rho_{S}(\theta) \le \rho_A(-\theta)$ and $\alpha=1$.
\end{theorem}

For $\theta \rightarrow 0$ the parameters $\rho_S(\theta)$ and $\rho_A(-\theta)$ approach the mean service time and mean inter-arrival time, respectively, so that $\rho_S(\theta) < \rho_A(-\theta)$ in Thm.~\ref{th:barrierexecutionmode} is the stability condition. Now, we can equate~\eqref{eq:bem-waitingtime-bound} with $\varepsilon$ and invert it to compute a bound on the $(1-\varepsilon)$-quantiles of the waiting time as
\begin{equation}
    \tau_{\varepsilon} = -\frac{1}{\theta}\ln \left( \frac{\varepsilon}{\alpha(\theta)} \right) .
    \label{eq:bem-waitingtime-quantile}
\end{equation}

We can derive a sojourn time bound from the waiting time bound using the fact that $T(n)=W(n)+Q(n)$.
For any fixed value of $k$, the CDF of the maximum of $k$ iid exponentials is
\begin{equation*}
    F_Q\left(\tau | k\right) = (1-e^{-\mu\tau})^k .
\end{equation*}
and by unconditioning on $k$ we obtain the general CDF
\begin{equation*}
    F_Q(\tau) = \sum_{k=1}^{k_{\max}}p(k)(1-e^{-\mu\tau})^k .
\end{equation*}
Note that for BEM jobs, $Q(n)$ is strictly the maximum service time of its $K(n)$ tasks.  For non-BEM jobs, this formulation makes the pessimistic assumption that all the tasks of the job are still executing when the final task starts.

So for $Q_i(n)$ exponential with parameter $\mu$ and a fixed value of $k$, $Q(n)$ has PDF
\begin{equation*}
    f_Q(\tau|k) = k(\mu e^{-\mu\tau})(1-e^{-\mu\tau})^{(k-1)} ,
\end{equation*}
With Thm.~\ref{th:barrierexecutionmode} $W(n)$ is bounded by
\begin{equation}
    F_W(\tau) \ge 1-\alpha e^{-\theta\tau}.
    \label{eq:bem-fw-cdf}
\end{equation}
Note that this can be negative for small $\tau$, so we add the stipulation that $F_W(\tau) = 0$ for $\tau < \frac{1}{\theta}\ln(\alpha)$.

Consider first the GI$\mid$M case when $\alpha=1$.  In this case we do not have to worry about~\eqref{eq:bem-fw-cdf} being negative, and we can compute the CDF of the sojourn time using the convolution
\begin{equation*}
    F_T(\tau)=\int_0^{\tau} F_W(\tau - x) f_Q(x) dx.
\end{equation*}
For any fixed $k$ this evaluates to
\begin{multline}
\label{eq:bem-tbound-FT1}
    F_T\left(\tau | k\right) \ge
    k \sum_{i=0}^{k-1} (-1)^i {{k-1}\choose{i}} \left[ \frac{1}{i+1}(1-e^{-(i+1)\mu\tau}) \right. \\
    \left. - \frac{\mu e^{-\theta\tau}}{(i+1)\mu-\theta}(1 - e^{-((i+1)\mu-\theta)\tau}) \right],
\end{multline}
and unconditioning gives us the sojourn time bound.
\begin{align}
    F_T\left(\tau\right) = \sum_{k=1}^{k_{\max}} p(k) F_T\left(\tau | k\right)
    \label{eq:ft-uncondition}
\end{align}

In the more general G$\mid$M case, we must take into account that $F_W(\tau) = 0$ for $\tau < \frac{1}{\theta}\ln(\alpha)$.  In this case the convolution becomes
\begin{equation*}
    F_T(\tau)=\int_0^{\tau_0} F_W(\tau - x) f_Q(x) dx ,
\end{equation*}
where $\tau_0=\tau - \frac{1}{\theta}\ln(\alpha)$.  For any fixed $k$ this works out to have some extra factors of $\alpha$:
\begin{multline}
    \label{eq:bem-t-bound-convoloution}
    \!\!\!\! F_T(\tau|k) \ge k \sum_{i=0}^{k-1} (-1)^i {{k-1}\choose{i}} \left[ \frac{1}{i+1}(1-\alpha^{\frac{(i+1)\mu}{\theta}}e^{-(i+1)\mu\tau}) \right. \\
    \left. - \frac{\alpha\mu e^{-\theta\tau}}{(i+1)\mu-\theta}(1 - \alpha^{\frac{(i+1)\mu}{\theta}-1}e^{-((i+1)\mu-\theta)\tau}) \right] .
\end{multline}
This can be unconditioned using~\eqref{eq:ft-uncondition} to give the final bound.

The resulting bound is slightly looser for small parallelism ratio, $k/s$, but numerically about the same as $k/s \rightarrow 1$. Note that both~\eqref{eq:bem-tbound-FT1} and~\eqref{eq:bem-t-bound-convoloution} add the constraint $\theta \neq (i+1)\mu$ for \mbox{$i\in \{0,\ldots,k_{\max}-1\}$}. Unfortunately there is no closed-form expression solving this inequality for $\tau$, so computing its solutions depends on a numerical search over the feasible~$\theta$.

\begin{figure*}
  \centering
  \includegraphics[width=0.9\columnwidth]{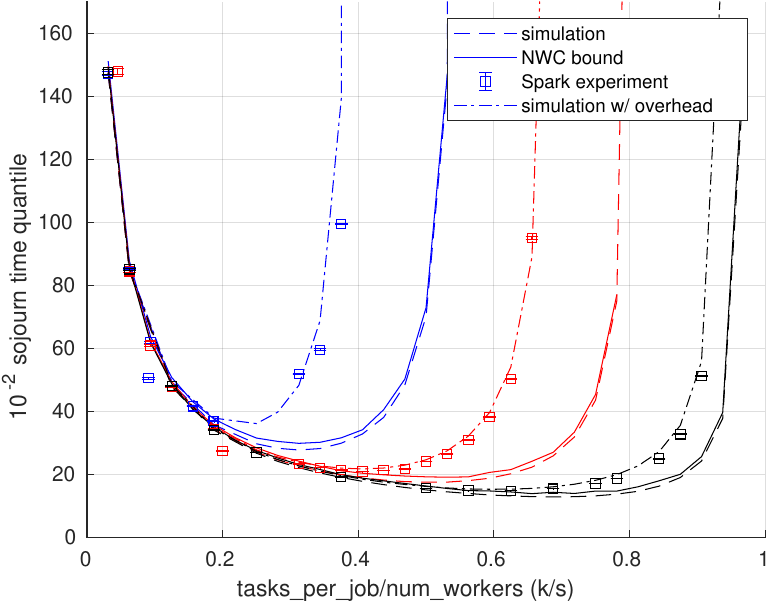} \ \ \ \ \
  \includegraphics[width=0.9\columnwidth]{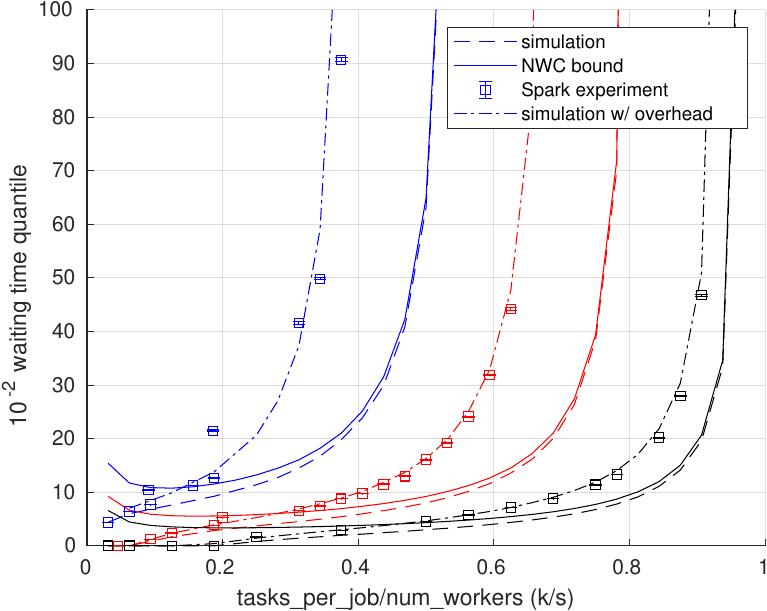}
  \caption{Sojourn and waiting time $10^{-2}$ quantiles for single-barrier BEM systems in Spark, compared to simulation and the analytical bounds, for utilization levels $\varrho\in\{0.3, 0.5, 0.7\}$.}
  \label{fig:bem-spark-sk-wait-sojourn}
\end{figure*}
%
%
\subsection{Results}
Fig.~\ref{fig:bem-hybrid} shows a plot of $10^{-6}$ quantiles of sojourn, waiting, and service time bounds for a BEM/non-BEM hybrid system with $s=32$ workers, $k=16$ tasks per job, with utilization of $\rho=0.7$, compared to simulation.  As expected, the service time bound is the same regardless of BEM/non-BEM fraction.  The waiting time, however, increases with increasing BEM fraction, driving a corresponding increase in sojourn time.

Fig.~\ref{fig:bem-hybrid-k} shows the same quantities for a purely BEM system with varying numbers of tasks per job and utilization of $\rho=0.4$.  Two examples are shown, both with a bimodal distribution for $K(n)$. The $x$-axis ranges over different fractions of small vs large-$k$ jobs.  We see that as the fraction of large-$k$ jobs increases, so do the waiting and sojourn times.  Also the increase is closer to linear in the fraction of large vs small jobs.  This is also expected. Both the distribution of BEM/non-BEM jobs and $K(n)$ influence the performance bound through the limits they put on $\theta$ in Thm.~\ref{th:barrierexecutionmode}, but the categorical distribution of $K(n)$ also manifests itself as the final $F_W$ and $F_T$ bounds become sums weighted by the PMF of $K(n)$.

The plots in Fig.~\ref{fig:bem-spark-sk-wait-sojourn} shows $10^{-2}$ quantiles of sojourn and waiting time bounds, along with simulation results, for purely non-hybrid systems, servicing only BEM jobs with fixed $k$ for utilization of $\varrho\in\{0.3, 0.5, 0.7\}$.  The $x$-axis used is the parallelism ratio, $k/s$, and the sojourn time results show the characteristic U-shape, with optimal performance near the center of the stability region.  The increase in sojourn time for small $k/s$ is due to longer execution times, as the jobs have less parallelism, and each task is required to do more work.  The increase as $k/s$ approaches the edge of the stability region is due to increased waiting time.  As predicted in section~\ref{sec:bem-1-stability} the stability region decreases for increasing $k$. The plots also contain data points from experiments on a stand-alone Spark cluster and simulations which include a model for system overhead, which will be discussed in section~\ref{sec:bem-spark-overhead}.
%
%
\section{BEM Spark Experiments and Overhead Model}
\label{sec:bem-spark-overhead}

\begin{figure*}
  \centering
  \includegraphics[width=0.9\columnwidth]{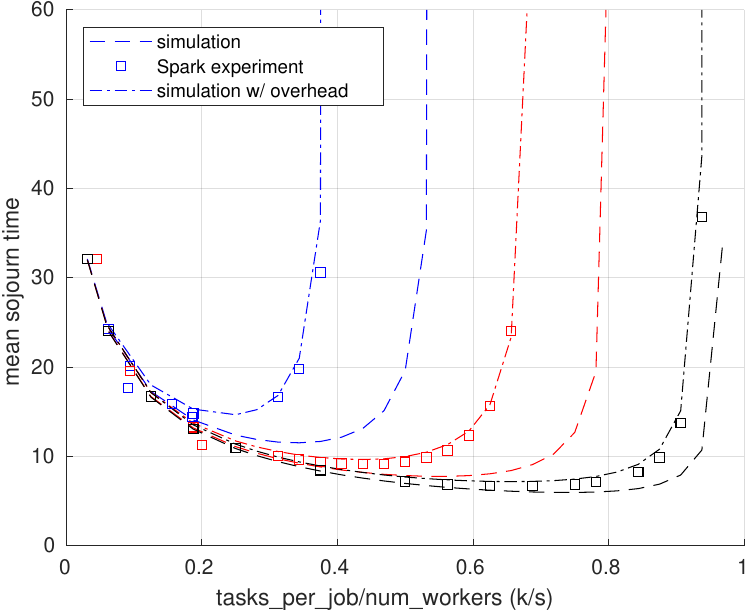} \ \ \ \ \
  \includegraphics[width=0.9\columnwidth]{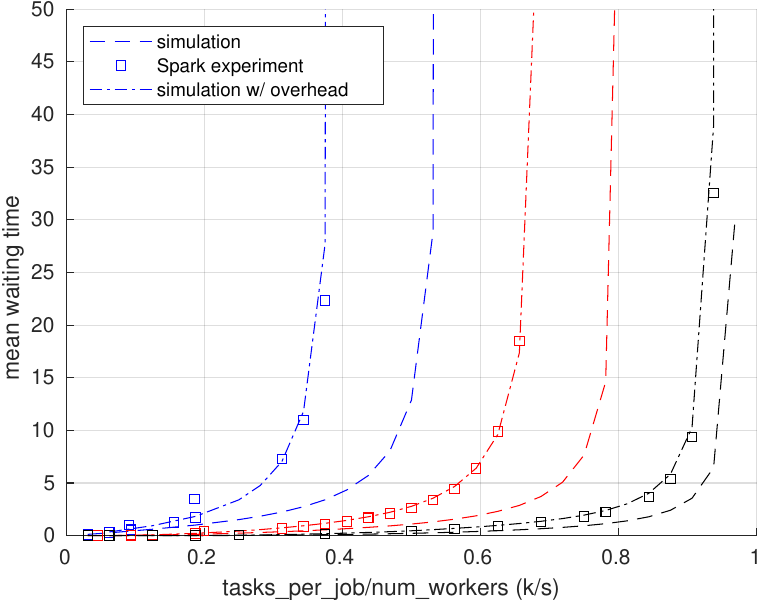}
  \caption{Sojourn and waiting time means for single-barrier BEM systems in Spark, compared to simulation and the analytical bounds, for utilization levels $\varrho\in\{0.3, 0.5, 0.7\}$.}
  \label{fig:bem-spark-sk-wait-sojourn-mean}
\end{figure*}

For the 1-barrier purely BEM system, we have performed comparable experiments on a standalone Spark cluster (v3.2.4).  In order to keep the test environment as controlled as possible we used $s=32$ single-core workers running in Docker containers on nodes in our Emulab testbed (now Cloudlab~\cite{cloudlab}).  Each of our physical servers has 16 logical cores, and hosted at most 12 of these Docker containers.
For these experiments we use an exponential arrival process with rate $\lambda$ and the tasks contain dummy workloads with execution time drawn from an exponential task service time distribution with mean $1/\mu = s/k$ seconds.  This ensures that with varying $k$, the rate of work arriving at the system is only a function of the arrival rate, and the system will have a utilization level of $\varrho = k \lambda/(s \mu)=\lambda$.  In the figures we plot the $10^{-2}$ quantiles rather than $10^{-6}$ because it would require orders of magnitude more data points, and these experiments are very time and energy intensive.

Figure~\ref{fig:bem-spark-sk-wait-sojourn} shows $10^{-2}$ quantiles, and Fig.~\ref{fig:bem-spark-sk-wait-sojourn-mean} shows means, of job sojourn and waiting times for these experiments for utilization levels of \mbox{$\varrho\in\{0.3, 0.5, 0.7\}$}.  The plots also include data from simulation and the analytical bounds presented in section~\ref{sec:bem-performance}.  We see that the waiting times in the real Spark system increase with increasing $k/s$, and that the sojourn times have the same characteristic U-shape.
The increase in sojourn time due to increasing execution time as $k/s$ gets smaller is exactly matched in the real system.  For larger $k/s$, however, we see that the waiting times in the real Spark system diverge more quickly than our simulation or model had predicted.  The stability regions of the system are correspondingly smaller, and the discrepancy becomes more dramatic with increasing load.

In order to understand the discrepancy, we thoroughly investigated the sources of overhead experienced by Spark jobs operating in BEM mode.
In our experiments, the actual execution time of each task can be controlled with precision on the order of 1~ms.  The sojourn time discrepancies observed, however are caused by overhead on the order of hundreds of ms.  In BEM mode all tasks of each job are executing at the same time, so there is no inter-task overhead to measure.  We consider as overhead anything that:
\begin{itemize}
    \item prevents a new job from starting after the the last task blocking it completes.
    \item prevents a job from departing (the JobFinished event in Spark) after its tasks are complete
\end{itemize}
We are interested in two general classes of overhead:
\begin{itemize}
    \item {\bf Blocking overhead} prevents new tasks/jobs from starting, when sufficient resources are available.
    \item {\bf Non-blocking overhead} is overhead that delays the departure of a job, even after its tasks have completed, but does not block subsequent jobs and tasks from using that job's workers.
\end{itemize}

An example of blocking overhead in Spark is task deserialization time, which is time spent transferring and unpacking the JARs and other task-related data on the workers.  This does not start until after the previous job departs and the worker is allocated the new job/task, and the new task cannot start until it is complete. Non-blocking overhead in Spark is more difficult to attribute to any particular source, since it is not explicitly included in any of Spark's logged statistics~\cite{bora-tiny-tasks-tpds}.  In trying to quantify it for BEM jobs, we found that it was subsumed by larger sources of overhead.

The largest source of overhead for Spark jobs in Barrier Execution Mode turns out to arise from the mechanism Spark uses to detect when enough workers are available to start the next job.  Each time a task finishes, or otherwise changes state, the \texttt{CoarseGrainedExecutorBackend} triggers a \texttt{statusUpdate()} in the driver's \texttt{CoarseGrainedSchedulerBackend}, which triggers the backend to make a single-executor resource offer to the scheduler.  This optimization is presumably to avoid iterating over all executors each time a task finishes, which would have time complexity of approximately $\mathcal{O}(s\cdot k)$ times the rate at which stages are processed, which could become cumbersome on extremely large clusters.  The backend also has a global version of \texttt{makeOffers()} which {\it does} iterate over all of the executors, and this is triggered when new stages are submitted (i.e. new jobs arrive) and periodically when the revive timer fires, which happens at $1000$ms intervals.

The fact that the scheduler is not offered a complete set of resources immediately when tasks finish, results in blocking overhead for all queued jobs, which dwarfs all other sources of overhead.  The distribution of this overhead has two parts.  Since the phase of the polling interval and the completion time of the tasks are effectively independent, this creates blocking overhead with a uniformly random distribution on the interval $[0,1]$ second.  Let $X_{timer} \sim \text{Unif}(0,1000)$ms.  Because the arrival of new jobs also triggers a global offer to the scheduler, there is a second component $X_{arrival}\sim \text{Exp}(\lambda)$ for exponential job arrivals.  The actual blocking overhead due to this waiting time $Y\sim \min(X_{timer}, X_{arrival})$.  Since they are independent, we can compute the distribution of the amount of time a job at the head of the queue has to wait, between when $k$ workers become available and when it is offered those workers, by taking the product of the two CCDFs:

\begin{align}
    \mathsf{P}(Y > y) &= \mathsf{P}(X_{timer} > y)\mathsf{P}(X_{arrival} > y) \\
    &= \frac{1000ms - y}{1000ms} e^{-\lambda y}  \ \ (0 \leq y \leq 1000ms)
\end{align}
And differentiating the CDF, $(1-\mathsf{P}(Y > y))$, gives the PDF.
\begin{align}
     f_Y(y) = \frac{1}{1000ms} e^{-\lambda y} \left( 1 + \lambda (1000ms - y) \right)
     \label{eq:spark-bem-overhead-pdf}
\end{align}

Fig.~\ref{fig:bem-spark-task-wait-overhead} shows a histogram of the time between $k$ workers becoming available and the first task of the next job starting for a Spark experiment with $s=32$, $k\in\{6,11\}$, $\varrho=0.7$, with the scaled PDF~\eqref{eq:spark-bem-overhead-pdf} overlaid for comparison.  We see that they match very well.

To validate this further, we implemented a blocking overhead model in our simulation, which adds to the service time of every task of every job.  In principle, all tasks are hit by the same revive timer and job arrival events.  But since the tasks' finish times are independent, the amount of blocking overhead each task experiences is independent.  There is also overhead from sources such as task deserialization, scheduling, and results gathering, and we have made extensive measurements of these, but those types of overhead are dwarfed in barrier execution mode by the time spent waiting for a global offer from the scheduler backend, so we do not make an effort to include them in our model.

We see in Figs.~\ref{fig:bem-spark-sk-wait-sojourn} and~\ref{fig:bem-spark-sk-wait-sojourn-mean} that the simulation with our overhead model matches both mean and quantiles of both sojourn and waiting times of the real system very closely.  Also, the effect of the waiting overhead, as a function of $k/s$, is more dramatic for higher utilization.  This is expected, since the overhead only affects jobs coming out of the queue, not those that are serviced immediately on arrival.

The blocking overhead we have described occurs due to the initial scheduling of stages in Spark.  Spark's BEM also supports placing synchronization points in the tasks via the \texttt{barrier()} or \texttt{allGather()} functions in the \texttt{BarrierTaskContext}.  Interestingly, a similar blocking overhead arises in this case, but from a different source, and the resulting overhead distribution is different.  In this case, when the tasks reach the synchronization point, they must wait for all other tasks to reach the same point before proceeding.  To accomplish this, the \texttt{BarrierTaskContext} class uses an abortable future to poll the \texttt{BarrierCoordinator} through the \texttt{RequestToSync()} method.  The future is checked from a loop containing a $1000$ms sleep.  Each task has its own \texttt{BarrierCoordinator}, so their polling timers are independent, and they are started when the task reaches the barrier.  In principle, the final task to reach the barrier does not have to wait, but the previous $k-1$ tasks will have to wait in the loop with frequency $1$Hz.  This means that all tasks, except possibly the last one to reach the barrier, have an additional blocking overhead with distribution $\text{Unif}(0,1)$ seconds.

\begin{figure}
  \centering
  \includegraphics[width=0.9\columnwidth]{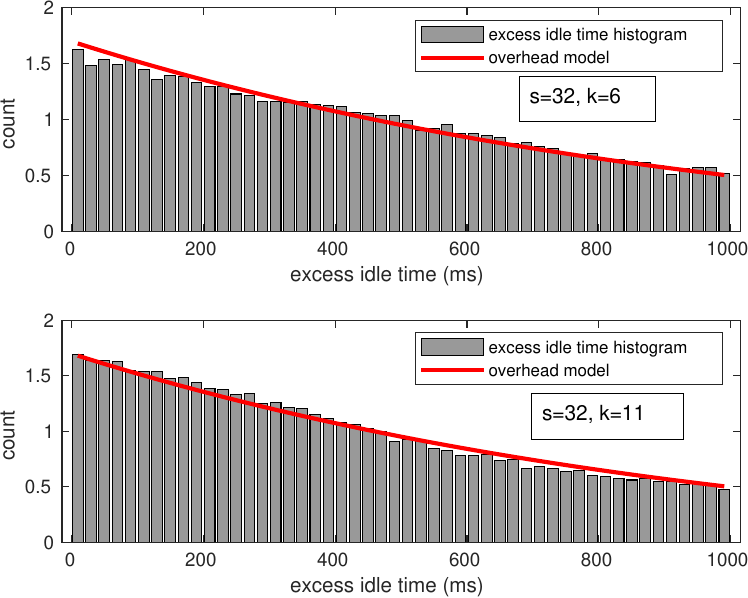}
  \caption{Distribution of idle times for barrier-mode jobs with $k=6$ and $k=11$ tasks queued on Spark cluster with $s=32$ workers.  This is the time, on a backlogged system, between when the system first has $k$ idle workers and when it begins servicing the next job.  The exponential shape is due to global offers triggered by the exponential job arrivals, and the truncation at $1000$ms is due to the polling timer.}
  \label{fig:bem-spark-task-wait-overhead}
\end{figure}
%
%
\section{Conclusions}
\label{sec:conclusion}
We have derived analytical expressions for the stability regions for parallel systems with blocking start and/or departure barriers, and also for $(s,k,l)$ 2-barrier systems, where jobs depart after only $l$ out of $k$ tasks finish.  We find that, even though some work is being discarded, the expected rate of useful work achieved by such a system can exceed that of the equivalent conventional BEM system.  In the case of 1-barrier $(s,k,l)$ systems we resorted to using a CTMC to model its behavior.  We also extended results from stochastic network calculus to derive waiting and sojourn time bounds for systems with blocking start barriers, including hybrid BEM systems which service a mix of barrier and non-barrier jobs with a random mix of different degrees of parallelism.  For pure BEM systems with fixed $k$, we found that for a given system utilization and number of servers, there is an optimal degree of parallelism that balances waiting time and job execution time.  Finally we performed extensive benchmarking experiments on a standalone spark cluster to validate the theoretical and simulation results. This led us to develop a model for the overhead introduced by how the Spark scheduler handles barrier mode jobs.
%
%
\balance
\bibliographystyle{IEEEtran}
\bibliography{IEEEabrv,ParallelSystems}
%
%
\end{document}